\documentclass[letterpaper, 10 pt, conference, twocolumn]{ieeeconf}

\IEEEoverridecommandlockouts                              

\usepackage{graphics} 
\usepackage{epsfig} 

\usepackage{times} 
\usepackage{amsmath} 
\usepackage{amssymb}  
\usepackage{dsfont}
\usepackage{booktabs}
\usepackage{mathtools}

\usepackage{amsthm}

\newtheorem{assumption}{Assumption}

\newtheorem{theorem}{Theorem}

\usepackage{nicefrac}       
\usepackage{microtype}      
\usepackage{tabto}
\usepackage{xcolor}

\usepackage{enumerate}
\usepackage{caption}
\usepackage{subcaption}
\usepackage{stfloats}

\usepackage[notocbasic]{nomencl}
\makenomenclature
\usepackage{etoolbox}

\usepackage{algorithm}
\usepackage[noend]{algorithmic}

\usepackage[colorlinks=true,linkcolor=blue]{hyperref}       
\usepackage{url}            
\usepackage{cite}
\usepackage{xpatch}
\xpatchcmd{\thenomenclature}{%
  \section*{\nomname}
}{
}{\typeout{Success}}{\typeout{Failure}}

\newcommand{\bbR}{\mathbb{R}}

\newcommand{\ra}{\rightarrow}

\newcommand{\rlbrack}[1]{\left [ {#1} \right]}

\newcommand{\rlpar}[1]{\left ( {#1} \right)}
\newcommand{\mean}[1]{\overline{{#1}}}

\title{\LARGE \bf
Best Response Convergence for Zero-sum Stochastic Dynamic Games with Partial and Asymmetric Information
}

\author{Yuxiang Guan \quad\quad  Iman Shames \quad\quad  Tyler H. Summers
\thanks{Y. Guan and T. Summers are with the Control, Optimization, and Networks Lab, University of Texas at Dallas. I. Shames is with the CIICADA Lab, Australian National University. This work was supported by the United States Air Force Office of Scientific Research under Grants FA9550-23-1-0424 and FA2386-24-1-4014, and by the National Science Foundation under Grant ECCS-2047040.}
}

\begin{document}

\maketitle
\thispagestyle{empty}
\pagestyle{empty}

\begin{abstract}
We analyze best response dynamics for finding a Nash equilibrium of an infinite horizon zero-sum stochastic linear quadratic dynamic game (LQDG) with partial and asymmetric information. We derive explicit expressions for each player's best response within the class of pure linear dynamic output feedback control strategies where the internal state dimension of each control strategy is an integer multiple of the system state dimension. With each best response, the players form increasingly higher-order belief states, leading to infinite-dimensional internal states. However, we observe in extensive numerical experiments that the game's value converges after just a few iterations, suggesting that strategies associated with increasingly higher-order belief states eventually provide no benefit. To help explain this convergence, our numerical analysis reveals rapid decay of the controllability and observability Gramian eigenvalues and Hankel singular values in higher-order belief dynamics, indicating that the higher-order belief dynamics become increasingly difficult for both players to control and observe. Consequently, the higher-order belief dynamics can be closely approximated by low-order belief dynamics with bounded error, and thus feedback strategies with limited internal state dimension can closely approximate a Nash equilibrium.
\end{abstract}

\section{Introduction}
Dynamic game theory provides a framework for analyzing and designing feedback strategies for decision-making agents that strategically interact in dynamic and uncertain environments \cite{basar1998}. Much of this theory assumes that all model parameters and the underlying system state are common knowledge for all agents, in which case dynamic programming techniques can be used to compute equilibrium feedback strategies. However, this assumption represents a major limitation of the theory, since in many settings agents must make decisions based on only partial information about the models and state.

Dynamic games where agents have only partial and asymmetric information of the state have been studied to a lesser extent \cite{tamer1973multistage,rhodes1969differential,zheng2013decomposition,gupta2014common,kartik2019zero,hambly2023linear}, including applications to robotics \cite{schwarting2021stochastic,peters2022learning}, auctions \cite{bichler2023learning}, and target defense and cybersecurity \cite{gupta2016dynamic,huang2020dynamic}. In these games, each player aims to control a common dynamic system to optimize their individual objective functions with distinct partial noisy measurements of the state. Fundamental challenges arise because the information asymmetry introduces \textit{belief representation} and \textit{theory of mind} issues, where players must impute the belief states and estimates of other players to inform their strategy. This leads to an infinite regress of higher-order beliefs, which renders even representing strategies and value functions difficult. Perspectives on partial/asymmetric information games have yielded groundbreaking work in other fields. Identification of crucial and surprising effects of information asymmetry in game theory has led to Nobel prize-winning work in information economics \cite{harsanyi1967games, akerlof1970market, aumann1976agreeing, spence1978job, stiglitz1981credit}.


Best response dynamics \cite{fudenberg1998theory,reeves2012computing} offer an approach to find pure Nash equilibria in the multi-player games \cite{harris1998rate,hofbauer2006best,sandholm2010population}. Each player begins with an initial strategy, then each player updates their strategy with a best response while keeping the strategies of the remaining players fixed. When this process converges, all players use mutual best response strategies, thus yielding a Nash equilibrium. 
Recently, best response methods have been applied to solve multi-player dynamic games in robotic applications \cite{williams2018best,wang2019game,wang2021game,schwarting2021stochastic}, among others. 

In the context of dynamic games with partial and asymmetric information, best response dynamics offer insights into the belief representation and theory of mind issues described above. In a two-player zero-sum linear quadratic Gaussian (LQG) dynamic game (where the state dynamics are linear, the cost function is quadratic, and each player receives only a noisy linear measurement of the state at each time) the following thought experiment illustrates an infinite regress of higher-order beliefs. Initializing player 1's strategy to the zero strategy, player 2's best response is an LQG controller, combining a linear quadratic regulator with a Kalman filter to estimate the underlying state from their output measurements. The controller has an internal dimension equal to the underlying state dimension, and the conditional state distribution represents a (0th-order) belief state. Now fixing this strategy, Player 1's best response is also an LQG controller in which the Kalman filter estimates not just the underlying state but also the the state estimate of Player 2 (a 1st-order belief state). With each successive best response, the players form increasingly higher-order belief states, leading to an infinite regress where the players' internal state dimensions increase towards infinity. It appears that in general a Nash equilibrium strategy pair must be infinite-dimensional.


In this work, we first derive explicit expressions for each player's best response within the class of pure linear dynamic output feedback control strategies where the internal state dimension of each control strategy is an integer multiple of the system state dimension. Despite the infinite regress of higher-order beliefs, our extensive numerical experiments show that the game's value converges after just a few iterations, suggesting that strategies associated with increasingly higher-order belief states eventually provide no benefit. We then analyze the controllability and observability Gramian eigenvalues and Hankel singular values of both players' higher-order belief dynamics, finding rapid decay in numerical experiments\textemdash a phenomenon also seen in network controllability studies \cite{pasqualetti2014controllability, ganapathy2021performance} and model reduction for large-scale systems \cite{antoulas2002decay}. Finally, we demonstrate that these decay rates can be bounded by Cholesky estimates, and show that higher-order belief dynamics can be closely approximated by low-order belief dynamics with bounded error. Thus, we find that feedback strategies with limited internal state dimension can closely approximate a Nash equilibrium.

The remainder of this paper is organized as follows. Section \ref{sec:problem_formulation} formulates an infinite horizon zero-sum stochastic LQDG with partial and asymmetric information. We derive explicit expressions for each player's best response in Section \ref{sec:iterated_best_response}. Section \ref{sec:control_observe_evaluation} analyzes the controllability and observability Gramians, Hankel singular values, and Cholesky estimates of these values. Section \ref{sec:experimental_results} presents our numerical experiments.

\section{Problem Formulation}\label{sec:problem_formulation}
We consider an infinite horizon two-player zero-sum stochastic LQDG with dynamics
\begin{align}\label{eqn:sys_orig}
    x_{t+1} = A x_t + B^1 u_t^1 + B^2 u_t^2 + w_t,
\end{align}
where $x_t \in \mathbb{R}^n$ is the system state with $x_0 \sim \mathcal{N}(\mean{x}_0, X_0)$, $u_t^1 \in \mathbb{R}^{m_1}$ is the control input for player 1, $u_t^2 \in \mathbb{R}^{m_2}$ is the control input for player 2, the system disturbance $w_t$ is an independent and identically distributed random vector with distribution $\mathcal{N}(0, W)$ with $W \in \mathbb{R}^{n \times n}$, and $A \in \mathbb{R}^{n \times n}$, $B^1 \in \mathbb{R}^{n \times m_1}$, and $B^2 \in \mathbb{R}^{n \times m_2}$ are system matrices.

To model game settings with partial and asymmetric information, the players do not have exact knowledge of the states and actions of other players. Instead, each player receives at time $t$ a signal $y_t^i$ through noisy sensor measurements, which is modeled by the following output equations
\begin{align}
    \begin{split}
        y_t^1 ={}& C^1 x_t + v_t^1, \\
        y_t^2 ={}& C^2 x_t + v_t^2,
    \end{split}
\end{align}
where $y_t^1 \in \mathbb{R}^{p_1}$ and $y_t^2 \in \mathbb{R}^{p_2}$ are the observed outputs of player 1 and player 2, $v_t^1$ and $v_t^2$ are independent and identically distributed random vectors with distributions $\mathcal{N}(0, V^1)$ and $\mathcal{N}(0, V^2)$ 
and $C^1 \in \mathbb{R}^{p_1 \times n}$ and $C^2 \in \mathbb{R}^{p_2 \times n}$ are output matrices.

Since the players do not know the exact value of the states, they must in general make decisions based on a history of all available information. Accordingly, we define the information states for $t = 0,1,...$
\begin{align}
    \begin{split}
        I_t^1 ={}& (y_0^1, y_1^1, \ldots, y_t^1, u_0^1, u_1^1, \ldots, u_{t-1}^1), \ I_0^1 = y_0^1, \\
        I_t^2 ={}& (y_0^2, y_1^2, \ldots, y_t^2, u_0^2, u_1^2, \ldots, u_{t-1}^2), \ I_0^2 = y_0^2.
    \end{split}
\end{align}
The players seek strategies defined by information feedback functions 
that optimize an objective function over an infinite horizon. More specifically, we consider a class of pure linear dynamic output feedback strategies where the internal state dimension of each control strategy is an integer multiple of the system state dimension, i.e., we seek $\pi$ with $u_t^1 = \pi(z_t^1)$ and $\mu$ with $u_t^2 = \mu(z_t^2)$ where $z_t^1 \in \mathbb{R}^{n_1 n}$ and $z_t^2 \in \mathbb{R}^{n_2 n}$ are states estimates (or belief states) with $n_1, n_2 \in \{1, 2, ... \}$ produced by filters of the form
\begin{align}\label{eqn:filter_form}
    \begin{split}
        z_{t+1}^1 ={}& A^1 z_t^1 + B^1 u_t^1 + L^1(y_t^1 - C^1 z_t^1), \\
        z_{t+1}^2 ={}& A^2 z_t^2 + B^2 u_t^2 + L^2(y_t^2 - C^2 z_t^2).
    \end{split}
\end{align}

For a zero-sum game, player $1$ (minimizer) minimizes the objective function while player $2$ (maximizer) maximizes. The objective function is
\begin{multline}\label{eqn:obj_fun}
        J = \lim_{T \ra \infty} \frac{1}{T}\mathbf{E}_{x_0, w_t, v_t^1, v_t^2} \Biggl[\sum_{t=0}^{T-1} x_t^\top Q x_t + \rlpar{u_t^1}^\top R^1 u_t^1 \\ + \rlpar{u_t^2}^\top R^2 u_t^2\Biggr],
\end{multline}
where $Q \in \mathbb{R}^{n \times n}$, $R^1 \in \mathbb{R}^{m_1 \times m_1}$, $R^2 \in \mathbb{R}^{m_2 \times m_2}$, $Q \succeq 0$, $R^1 \succ 0$, and $R^2 \prec 0$. We assume that the penalty $R^2$ on player $2$'s input is sufficiently large so that the upper value of the game is bounded. We are interested in finding a saddle point (Nash) equilibrium policy pair ($\pi^*, \mu^*$) where
\begin{align}
    J(\pi^*, \mu) \leq J(\pi^*, \mu^*) \leq J(\pi, \mu^*), \ \forall \pi, \mu,
\end{align}
that is, no player can unilaterally improve their cost by deviating from their equilibrium policy.
\begin{assumption}\label{asmp:asmp1}
    We assume that all system parameters (comprising the system matrices $A$, $B^1$, $B^2$, disturbance covariance $W$, initial state mean $\overline{x}_0$ and covariance $X_0$, output matrices $C^1$ and $C^2$, and measurement noise covariances $V^1$ and $V^2$) and all cost parameters (comprising the state cost $Q$, minimizer's input cost $R^1$, and maximizer's input cost $R^2$) are common knowledge across all players.
\end{assumption}
Assumption \ref{asmp:asmp1} allows each player to independently reason about their best responses and those of the other player to determine equilibrium strategies. 



\section{Best Response Dynamics}\label{sec:iterated_best_response}
In this section, we analyze best response dynamics for solving a two-player zero-sum dynamic game with partial and asymmetric information. We derive expressions for the best response for each player. To help understand convergence properties of the value of the game, we evaluate controllability and observability metrics of each player's belief state dynamics in Section \ref{sec:control_observe_evaluation}.

\subsection{LQG Control of the Minimizer} \label{LQG-Min-Iteration1}
We first initialize a zero strategy for the maximizer (where $u_t^2 = 0$ $\forall t$). Then the minimizer solves a standard LQG control problem with dynamics and measurement
\begin{align}\label{eqn:sys_min}
    \begin{split}
        x_{t+1} ={}& A x_t + B^1 u_t^1 + w_t, \\
        y_t^1 ={}& C^1 x_t + v_t^1,
    \end{split}
\end{align}
and objective
\begin{align*}
    J^1  = \lim_{T \ra \infty} \frac{1}{T} \mathbf{E}_{x_0, w_t, v_t^1} \rlbrack{\sum_{t=0}^{T-1} x_t^\top Q x_t + \rlpar{u_t^1}^\top R^1 u_t^1}.
\end{align*}

The minimizer employs a steady-state Kalman filter to form an estimate $z_t^1 \in \mathbb{R}^n$ of the state given by
\begin{align*}
    z_{t+1}^1 &= A z_t^1 + B^1 u_t^1 + L^1 \rlpar{y_t^1 - C^1 z_t^1}, \\
    L^{1} &= A \Sigma^{1} \rlpar{C^1}^\top \rlpar{V^1 + C^1 \Sigma^{1} \rlpar{C^1}^\top}^{-1},
\end{align*}
where the optimal estimator gain $L^{1}$ utilizes the steady-state estimation error covariance matrix $\Sigma^1$ that solves the algebraic Riccati equation
\begin{align}\label{eqn:Sigma1}
    \begin{split}
        \Sigma^{1} ={}& W + A \Sigma^{1} A^\top - A \Sigma^{1} \rlpar{C^1}^\top \\
        & \rlpar{V^1 + C^1 \Sigma^{1} \rlpar{C^1}^\top}^{-1} C^1 \Sigma^{1} A^\top.
    \end{split}
\end{align}
The steady-state estimation error covariance is given by $\Sigma^1 = \lim_{t\rightarrow \infty} \Sigma_t^1$, where $\Sigma_t^1 = \mathbf{E} [ e_t^1 (e_t^1)^\top]$, where the estimation error $e_t^1 \coloneqq x_t - z_t^1 $ has dynamics
\begin{align*}
    \begin{split}
        e_{t+1}^1 ={}& x_{t+1} - z_{t+1}^1 \\
        ={}& \rlpar{A - L^1 C^1} e_t^1 + w_t - L^1 v_t^1.
    \end{split}
\end{align*}

The optimal strategy is the linear state estimate feedback 
\begin{align}\label{eqn:u_t_1}
    u^{1}_t = K^{1} z_t^1,
\end{align}
where the gain matrix $K^1 \in \mathbb{R}^{m_1 \times n}$ is given by
\begin{align}\label{eqn:K1}
    K^{1} = -\rlpar{R^1 + \rlpar{B^1}^\top P^{1} B^1}^{-1} \rlpar{B^1}^\top P^{1} A,
\end{align}
where the cost matrix $P^{1}$ solves the algebraic Riccati equation
\begin{align}\label{eqn:P1}
      \begin{split}
        P^{1} ={}& Q + A^\top P^{1} A - A^\top P^{1} B^1 \\
            & \rlpar{R^1 + \rlpar{B^1}^\top P^{1} B^1}^{-1} \rlpar{B^1}^\top P^{1} A.
    \end{split}
\end{align}
There exist unique solutions $P^1$ and $\Sigma^1$ of the Riccati equations if the system \eqref{eqn:sys_min} is controllable and observable. 

The closed-loop state and state estimate dynamics are
{\small
\begin{align*} \nonumber
     \begin{bmatrix}
         x_{t+1} \\ z_{t+1}^1
     \end{bmatrix} &= \underbrace{\begin{bmatrix}
        A & B^1 K^{1} \\
        L^{1} C^1 & A + B^1 K^{1} - L^{1} C^1
    \end{bmatrix}}_{:= \bar A} \begin{bmatrix}
         x_{t} \\ z_{t}^1
     \end{bmatrix} + \underbrace{\begin{bmatrix}
        I & 0 \\
        0 & L^1
    \end{bmatrix}}_{F} \begin{bmatrix}
         w_{t} \\ v_{t}^1
     \end{bmatrix}.
\end{align*}
}
The optimal cost of the minimizer is then
\begin{align*}
    J^{1*} = Tr \rlpar{\bar P \bar W} =  Tr \rlpar{\bar \Sigma \bar Q},
\end{align*}
where $$\bar W =  F \begin{bmatrix}
        W & 0 \\
        0 & V^1
    \end{bmatrix} F^\top, \quad \bar Q =  \begin{bmatrix}
        Q & 0 \\
        0 & \rlpar{K^{1}}^\top R^1 K^{1}
    \end{bmatrix},$$ and $\bar P$ and $\bar \Sigma$ solve the respective Lyapunov equations
\begin{align*}
    \bar P &= \bar A^\top \bar P \bar A + \bar Q, \quad\quad \bar \Sigma = \bar A \bar \Sigma \bar A^\top + \bar W.
\end{align*}

\subsection{Best Response from the Maximizer}\label{subsec:BR_Max}
Having analyzed the minimizer's optimal strategy in the previous section, we now turn to finding the best response for the maximizer with the minimizer's strategy fixed. The dynamic system now faced by the maximizer incorporates both the state and the minimizer's state estimate (0th-order belief). We define an augmented system
\begin{align}\label{eqn:sys_max}
    X_{t+1}^2 = \overline{A}^2 X_t^2 + \overline{B}^2 u_t^2 + F^2 \Tilde{w}_t^2,
\end{align}
with augmented state $X_t^2 \coloneqq [x_t^\top \quad (z_t^1)^\top]^\top \in \mathbb{R}^{2n}$, where $\Tilde{w}_t^2 = [w_t^\top \quad (v_t^1)^\top]^\top$ and
\begin{align*}
    &\overline{A}^2 = 
    \begin{bmatrix}
        A & B^1 K^{1} \\
        L^{1} C^1 & A + B^1 K^{1} - L^{1} C^1
    \end{bmatrix}, \
    \overline{B}^2 = 
    \begin{bmatrix}
        B^2 \\
        0
    \end{bmatrix}, \\
    &F^2 = 
    \begin{bmatrix}
        I & 0 \\
        0 & L^{1}
    \end{bmatrix}, \ \text{and} \
    \overline{W}^2 = 
    \mathbf{E} \rlbrack{\Tilde{w}_t^2 \rlpar{\Tilde{w}_t^2}^\top} = 
    \begin{bmatrix}
        W & 0 \\
        0 & V^1
    \end{bmatrix}.
\end{align*}
The maximizer's measurement is reformulated as
\begin{align}\label{eqn:sys_max_output}
    y_t^2 = \overline{C}^2 X_t^2 + v_t^2,
\end{align}
where $\overline{C}^2 = [C^2 \quad 0]$. Given the optimal strategy \eqref{eqn:u_t_1} of the minimizer, we reformulate the objective function \eqref{eqn:obj_fun} of the maximizer as
\begin{equation} \nonumber
    J^2 = \lim_{T \ra \infty} \frac{1}{T} \mathbf{E}_{X_0^2, \Tilde{w}_t^2, v_t^2} \Biggl[\sum_{t=0}^{T-1} \rlpar{X_t^2}^\top \overline{Q}^2 X_t^2 + \rlpar{u_t^2}^\top R^2 u_t^2 \Biggr],
\end{equation}
where
\begin{align*}
    \overline{Q}^2 = 
    \begin{bmatrix}
        Q & 0 \\
        0 & \rlpar{K^{1}}^\top R^1 K^{1}
    \end{bmatrix}.
\end{align*}

The maximizer's best response involves employing a Kalman filter to estimate the \emph{augmented} state, which requires estimating \emph{both} the system state and the state estimate of the minimizer. Thus, the maximizer constructs a 1st-order belief (its belief about the minimizer's belief about the state). The maximizer's belief state $z_t^2 \in \mathbb{R}^{2n}$ is obtained through a state estimator of the form
\begin{align*}
    z_{t+1}^2 = \overline{A}^2 z_t^2 + \overline{B}^2 u_t^2 + L^2(y_t^2 - \overline{C}^2 z_t^2).
\end{align*}
Defining the estimation error $e_t^2 \coloneqq X_t^2 - z_t^2$, the estimation error dynamics are
\begin{align*}
    e_{t+1}^2 ={}& X_{t+1}^2 - z_{t+1}^2 \\
                       ={}& \rlpar{\overline{A}^2 - L^2 \overline{C}^2} e_t^2 + F^2 \Tilde{w}_t^2 - L^2 v_t^2.
\end{align*}    
Again using standard Kalman filtering results, the optimal steady-state estimator gain is 
\begin{align}\label{eqn:L2}
    L^{2} = \overline{A}^2 \Sigma^{2} \rlpar{\overline{C}^2}^\top \rlpar{V^2 + \overline{C}^2 \Sigma^{2} \rlpar{\overline{C}^2}^\top}^{-1},
\end{align}
where the steady-state error covariance matrix $\Sigma^{2} = \lim_{t \rightarrow \infty}\mathbf{E}[e_t^2 (e_t^2)^\top] $ solves the algebraic Riccati equation
\begin{align}\label{eqn:Sigma2}
    \begin{split}
        \Sigma^{2} ={}& F^2 \overline{W}^2 \rlpar{F^2}^\top + \overline{A}^2 \Sigma^{2} \rlpar{\overline{A}^2}^\top - \overline{A}^2 \Sigma^{2} \rlpar{\overline{C}^2}^\top \\
        & \rlpar{V^2 + \overline{C}^2 \Sigma^2 \rlpar{\overline{C}^2}^\top}^{-1} \overline{C}^2 \Sigma^{2} \rlpar{\overline{A}^2}^\top.
    \end{split}
\end{align}
This equation has a unique solution if the \emph{augmented} system \eqref{eqn:sys_max}, \eqref{eqn:sys_max_output} is observable. 

The best response strategy of the maximizer is the linear (augmented/belief) state estimate feedback
\begin{align*}
    u^{2}_t = K^{2} z^2_t,
\end{align*}
where the feedback gain matrix $K^{2} \in \mathbb{R}^{m_2 \times 2n}$ is given by 
\begin{align}\label{eqn:K2}
    K^{2} = -\rlpar{R^2 + \rlpar{\overline{B}^2}^\top P^{2} \overline{B}^2}^{-1} \rlpar{\overline{B}^2}^\top P^{2} \overline{A}^2,
\end{align}
and the cost matrix $P^{2}$ solves the algebraic Riccati equation
\begin{align}\label{eqn:P2}
    \begin{split}
        P^{2} ={}& \overline{Q}^2 + \rlpar{\overline{A}^2}^\top P^{2} \overline{A}^2 - \rlpar{\overline{A}^2}^\top P^{2} \overline{B}^2 \\ 
        &\rlpar{R^2 + \rlpar{\overline{B}^2}^\top P^{2} \overline{B}^2}^{-1} \rlpar{\overline{B}^2}^\top P^{2} \overline{A}^2.
    \end{split}
\end{align}
Importantly, we emphasize that maximizer's optimal cost is bounded only when a solution exists where $R^2 + \rlpar{\overline{B}^2}^\top P^{2} \overline{B}^2 \prec 0$, which requires that the maximizer input penalty $R^2$ is sufficiently negative definite. 

The closed-loop state and belief state dynamics are
{\small
\begin{align*} \nonumber
     \begin{bmatrix}
         X^2_{t+1} \\ z_{t+1}^2
     \end{bmatrix} &= \underbrace{\begin{bmatrix}
        \overline{A}^2 & \overline{B}^2 K^{2} \\
        L^{2} \overline{C}^2 & \overline{A}^2 + \overline{B}^2 K^{2} - L^{2} \overline{C}^2
    \end{bmatrix}}_{:= \tilde A} \begin{bmatrix}
         X^2_{t} \\ z_{t}^2
     \end{bmatrix} \\ &+ \underbrace{\begin{bmatrix}
        F^2 & 0 \\
        0 & L^2
    \end{bmatrix}}_{\tilde F} \begin{bmatrix}
         \tilde w^2_{t} \\ v_{t}^2
     \end{bmatrix}.
\end{align*}
}
The optimal cost of the maximizer is then
\begin{align*}
    J^{2*} = Tr \rlpar{\tilde P \tilde W} =  Tr \rlpar{\tilde \Sigma \tilde Q},
\end{align*}
where $$\tilde W =  \tilde F \begin{bmatrix}
        \overline{W}^2 & 0 \\
        0 & V^2
    \end{bmatrix} \tilde F^\top, \quad \tilde Q =  \begin{bmatrix}
        \overline{Q}^2 & 0 \\
        0 & \rlpar{K^{2}}^\top R^2 K^{2}
    \end{bmatrix},$$ and $\tilde P$ and $\tilde \Sigma$ solve the respective Lyapunov equations
\begin{align*}
    \tilde P &= \tilde A^\top \tilde P \tilde A + \tilde Q, \quad\quad \tilde \Sigma = \tilde A \tilde \Sigma \tilde A^\top + \tilde W. 
\end{align*}


\subsection{Best Response Dynamics}\label{sec:best_response_dynamics}
We now describe explicit expressions for general best response dynamics that, upon convergence, yield a Nash equilibrium within our pure dynamic output feedback strategy class. 
Below we summarize the best response dynamics for each player for our dynamic game with partial and asymmetric information. In each iteration, the opposing player's strategy is fixed, and each player faces an LQG control problem, but for an augmented system that includes both the system state and increasingly higher-order belief states of the opposing player. Each iteration increases the dimension of the state estimate for the augmented belief state system to an incremental integer multiple of the original state dimension. This reformulation captures recursive definitions for the augmented belief state dynamics, objective function, and state estimator. The process iterates until both players' best responses converge. If the iterations converge, we obtain a Nash equilibrium within our strategy class.

The expressions for the best response dynamics below follow from recursively solving an LQG problem for each player.
By construction, augmented belief dynamics considered by each player at each iteration increase in dimension at each best response iteration, thus increasing the internal dimensions of the feedback strategies employed by each player.

\subsubsection{Minimizer's Best Response Characteristics}\label{prop1}

    The minimizer's best response at iteration $k=2,3,\ldots$ considers the augmented belief dynamics and measurement 
    \begin{align}\label{eqn:min_aug_sys}
        \begin{split}
            X_{t+1,k}^1 ={}& \overline{A}_k^1 X_{t,k}^1 + \overline{B}_k^1 u_{t,k}^1 + F_k^1 \Tilde{w}_{t}^1, \\
            y_{t,k}^1 ={}& \overline{C}_k^1 X_{t,k}^1 + v_{t}^1,
        \end{split}
    \end{align}
    where $X_{t,k}^1 \coloneqq [x_{t}^\top \quad (z_{t,k-1}^2)^\top]^\top \in \mathbb{R}^{(2k-1)n}$, $\Tilde{w}_{t}^1 = [w_t^\top \quad (v_t^2)^\top]^\top$,
    \begin{align*}
        \overline{A}_{k}^1 =
        \begin{bmatrix}
            A & B^2 K_{k-1}^{2} \\
            L_{k-1}^{2} C^2 & \overline{A}_{k-1}^2 + \overline{B}_{k-1}^2 K_{k-1}^2 - L_{k-1}^2 \overline{C}_{k-1}^2
        \end{bmatrix}, \ 
    \end{align*}
    \begin{align*}
         \overline{B}_k^1 =
        \begin{bmatrix}
            B^1 \\
            0
        \end{bmatrix}, \
        \overline{C}_k^1 = 
        \begin{bmatrix}
            C^1 & 0
        \end{bmatrix}, \
        F_{k}^1 = 
        \begin{bmatrix}
            I & 0 \\
            0 & L_{k-1}^{2}
        \end{bmatrix}, \
    \end{align*}
    \begin{align*}
        \overline{W}^1 = \mathbf{E} \rlbrack{\Tilde{w}_{t}^1 \rlpar{\Tilde{w}_{t}^1}^\top} =
        \begin{bmatrix}
            W & 0 \\
            0 & V^2 
        \end{bmatrix},
    \end{align*}
    with $\overline{A}_k^1 \in \bbR^{(2k-1)n \times (2k-1)n}$, $\overline{B}_k^1 \in \bbR^{(2k-1)n \times m_1}$, $\overline{C}_k^1 \in \bbR^{p_1 \times (2k-1)n}$, $F_k^1 \in \bbR^{(2k-1)n \times n+p_2}$, and $\overline{W}^1 \in \bbR^{n+p_2 \times n+p_2}$. Moreover, $K_{k-1}^2$, $L_{k-1}^2$, $\overline{A}_{k-1}^2$, $\overline{B}_{k-1}^2$, and $\overline{C}_{k-1}^2$ are as described in Section \ref{prop2} with $K_1^2 = K^2$, $L_1^2 = L^2$, $\overline{A}_1^2 = \overline{A}^2$, $\overline{B}_1^2 = \overline{B}^2$, $\overline{C}_1^2 = \overline{C}^2$ from the maximizer's initial response from Section \ref{subsec:BR_Max}. 
    The objective function of the minimizer at iteration $k$ is 
    {\small
    \begin{equation} \nonumber
        J_k^1 = \lim_{T \ra \infty} \frac{1}{T} \mathbf{E}_{X_{0,k}^1, \Tilde{w}_{t}^1,v_t^1} \Biggl[\sum_{t=0}^{T-1} X_{t,k}^{1\top} \overline{Q}_k^1 X_{t,k}^1  + u_{t,k}^{1\top} R^1 u_{t,k}^1 \Biggr],
    \end{equation}
    }
    where
    \begin{align*}
        \overline{Q}_k^1 = 
        \begin{bmatrix}
            Q & 0 \\
            0 & \rlpar{K_{k-1}^{2}}^\top R^2 K_{k-1}^{2}
        \end{bmatrix} \in \bbR^{(2k-1)n \times (2k-1)n}.
    \end{align*}
    
    The minimizer's belief state $z_{t,k}^1 \in \mathbb{R}^{(2k-1)n}$ at iteration $k$ is obtained through a state estimator and feedback strategy
    \begin{align}
        z_{t+1,k}^1 &= \overline{A}_k^1 z_{t,k}^1 + \overline{B}_k^1 u_{t,k}^1 + L_{k}^1 (y_{t,k}^1 - \overline{C}_k^1 z_{t,k}^1), \\
        u_{t,k}^1 &= K_k^1 z_{t,k}^1.
    \end{align}
    The optimal feedback gain $K_k^1 \in \bbR^{m_1 \times (2k-1)n}$, state estimator gain $L_k^1 \in \bbR^{(2k-1)n \times p_1}$, associated cost parameter $P_k^1 \in \bbR^{(2k-1)n \times (2k-1)n}$, and error covariance matrix $\Sigma_k^1 \in \bbR^{(2k-1)n \times (2k-1)n}$ are given by
    {\small
    \begin{align}
        K_k^{1} ={}& -\rlpar{R^1 + \rlpar{\overline{B}_k^1}^\top P_k^{1} \overline{B}_k^1}^{-1} \rlpar{\overline{B}_k^1}^\top P_k^{1} \overline{A}_k^1, \label{eqn:K1_k} \\
        \begin{split}\label{eqn:P1_k}
            P_k^{1} ={}& \overline{Q}_k^1 + \rlpar{\overline{A}_k^1}^\top P_k^{1} \overline{A}_k^1 - \rlpar{\overline{A}_k^1}^\top P_k^{1} \overline{B}_k^1 \\ & \rlpar{R^1 + \rlpar{\overline{B}_k^1}^\top P_k^{1} \overline{B}_k^1}^{-1} \rlpar{\overline{B}_k^1}^\top P_k^{1} \overline{A}_k^1,
        \end{split} \\
        L_k^{1} ={}& \overline{A}_k^1 \Sigma_k^{1} \rlpar{\overline{C}_k^1}^\top \rlpar{V^1 + \overline{C}_k^1 \Sigma_k^{1} \rlpar{\overline{C}_k^1}^\top}^{-1}, \label{eqn:L1_k}\\
        \begin{split}\label{eqn:Sigma1_k}
            \Sigma_k^{1} ={}& F_k^1 \overline{W}^1 \rlpar{F_k^1}^\top + \overline{A}_k^1 \Sigma_k^{1} \rlpar{\overline{A}_k^1}^\top - \overline{A}_k^1 \Sigma_k^{1} \rlpar{\overline{C}_k^1}^\top \\
            & \rlpar{V^1 + \overline{C}_k^1 \Sigma_k^1 \rlpar{\overline{C}_k^1}^\top}^{-1} \overline{C}_k^1 \Sigma_k^{1} \rlpar{\overline{A}_k^1}^\top.
        \end{split}
    \end{align}}
    Although the cost weighting matrix $\overline{Q}_k^1$ is indefinite since $Q \succeq 0$ and $(K_{k-1}^{2})^\top R^2 K_{k-1}^{2} \prec 0$, unique solutions $P_k^1$ and $\Sigma_k^1$ for the Riccati equations still exist under certain conditions, provided the system \eqref{eqn:min_aug_sys} is stabilizable and detectable. Also, the closed-loop system, under the state feedback strategy $u_{t,k}^1 = K_k^1 z_{t,k}^1$, remains stable at each iteration \cite{willemsLeast1971,molinariStable1973a}.
The closed-loop state and belief state dynamics are
    {\small
    \begin{align*} \nonumber
         \begin{bmatrix}
             X^1_{t+1,k} \\ z_{t+1,k}^1
         \end{bmatrix} &= \underbrace{\begin{bmatrix}
            \overline{A}_k^1 & \overline{B}_k^1 K_k^{1} \\
            L_k^{1} \overline{C}_k^1 & \overline{A}_k^1 + \overline{B}_k^1 K_k^{1} - L_k^{1} \overline{C}_k^1
        \end{bmatrix}}_{:= \tilde A_k^1} \begin{bmatrix}
             X^1_{t,k} \\ z_{t,k}^1
         \end{bmatrix} \\ &+ \underbrace{\begin{bmatrix}
            F_k^1 & 0 \\
            0 & L_k^1
        \end{bmatrix}}_{\tilde F_k^1} \begin{bmatrix}
             \tilde w^1_{t} \\ v_{t}^1
         \end{bmatrix}.
    \end{align*}
    }
    The optimal cost of the minimizer at each iteration is then
    \begin{align*}
        J_k^{1*} = Tr \rlpar{\tilde P_k^1 \tilde W_k^1} =  Tr \rlpar{\tilde \Sigma_k^1 \tilde Q_k^1},
    \end{align*}
    where {\small$$\tilde W_k^1 =  \tilde F_k^1 \begin{bmatrix}
        \overline{W}^1 & 0 \\
        0 & V^1
    \end{bmatrix} \rlpar{\tilde F_k^1}^\top, \ \tilde Q_k^1 =  \begin{bmatrix}
        \overline{Q}_k^1 & 0 \\
        0 & \rlpar{K_k^{1}}^\top R^1 K_k^{1}
    \end{bmatrix},$$} and $\tilde P_k^1$ and $\tilde \Sigma_k^1$ solve the respective Lyapunov equations
    {\small\begin{align*}
        \tilde P_k^1 &= \rlpar{\tilde A_k^1}^\top \tilde P_k^1 \rlpar{\tilde A_k^1} + \tilde Q_k^1, \ \tilde \Sigma_k^1 = \rlpar{\tilde A_k^1} \tilde \Sigma_k^1 \rlpar{\tilde A_k^1}^\top + \tilde W_k^1. 
    \end{align*}}

\subsubsection{Maximizer's Best Response Characteristics} \label{prop2}
    The maximizer's best response at iteration $k=1,2,3,\ldots$ considers the augmented belief dynamics and measurement
    \begin{align}\label{eqn:max_aug_sys}
        \begin{split}
            X_{t+1,k}^2 ={}& \overline{A}_k^2 X_{t,k}^2 + \overline{B}_k^2 u_{t,k}^2 + F_k^2 \Tilde{w}_{t}^2, \\
            y_{t,k}^2 ={}& \overline{C}_k^2 X_{t,k}^2 + v_{t}^2,
        \end{split}
    \end{align}
    where $X_{t,k}^2 \coloneqq [x_{t}^\top \quad (z_{t,k}^1)^\top]^\top \in \mathbb{R}^{2kn}$, $\Tilde{w}_{t}^2 = [w_t^\top \quad (v_t^1)^\top]^\top$,
    \begin{align*}
        \overline{A}_{k}^2 =
        \begin{bmatrix}
            A & \overline{B}^1 K_k^{1} \\
            L_{k}^{1} \overline{C}_k^1 & \overline{A}_{k}^1 + \overline{B}_{k}^1 K_{k}^1 - L_{k}^1 \overline{C}_{k}^1
        \end{bmatrix}, \ 
        \overline{B}_k^2 =
        \begin{bmatrix}
            B^2 \\
            0
        \end{bmatrix},
    \end{align*}
    \begin{align*}
        \overline{C}_k^2 = 
        \begin{bmatrix}
            C^2 & 0
        \end{bmatrix}, \
        F_{k}^2 = 
        \begin{bmatrix}
            I & 0 \\
            0 & L_k^{1}
        \end{bmatrix}, \
        \overline{W}^2 =
        \begin{bmatrix}
            W & 0 \\
            0 & V^1 
        \end{bmatrix},
    \end{align*}
    with $\overline{A}_k^2 \in \bbR^{2kn \times 2kn}$, $\overline{B}_k^2 \in \bbR^{2kn \times m_2}$, $\overline{C}_k^2 \in \bbR^{p_2 \times 2kn}$, $F_k^2 \in \bbR^{2kn \times n+p_1}$, and $\overline{W}^2 \in \bbR^{n+p_1 \times n+p_1}$. Furthermore, the matrices $K_k^1$, $L^1_k$, $\overline{A}_{k}^1$, $\overline{B}_{k}^1$, $\overline{C}_{k}^1$ are presented in Section \ref{prop1} with initial values $K_1^1 = K^1$, $L_1^1 = L^1$, $\overline{A}_1^1 = A$, $\overline{B}_1^1 = B^1$, $\overline{C}_1^1 = C^1$ from Section \ref{LQG-Min-Iteration1}.
    The objective function of the maximizer at iteration $k$ is 
    {\small
    \begin{equation}
        J_k^2 = \lim_{T \ra \infty} \frac{1}{T} \mathbf{E}_{X_{0,k}^2, \Tilde{w}_{t}^2,v_t^2} \Biggl[\sum_{t=0}^{T-1} \rlpar{X_{t,k}^2}^\top \overline{Q}_k^2 X_{t,k}^2 + \rlpar{u_{t,k}^2}^\top R^2 u_{t,k}^2 \Biggr],
    \end{equation}
    }
    where
    \begin{align*}
        \overline{Q}_k^2 = 
        \begin{bmatrix}
            Q & 0 \\
            0 & \rlpar{K_k^{1}}^\top R^1 K_k^{1}
        \end{bmatrix} \in \bbR^{2kn \times 2kn}.
    \end{align*}
    The maximizer's belief state $z_{t,k}^2 \in \mathbb{R}^{2kn}$ at iteration $k$ is obtained through a state estimator and feedback strategy
    \begin{align}
        z_{t+1,k}^2 &= \overline{A}_k^2 z_{t,k}^2 + \overline{B}_k^2 u_{t,k}^2 + L_{k}^2 (y_{t,k}^2 - \overline{C}_k^2 z_{t,k}^2),\\
        u_{t,k}^2 &= K_k^{2} z_{t,k}^2.
    \end{align}
    The optimal feedback gain $K_k^2 \in \bbR^{m_2 \times 2kn}$, state estimator gain $L_k^2 \in \bbR^{2kn \times p_2}$, associated cost parameter $P_k^2 \in \bbR^{2kn \times 2kn}$, and error covariance matrix $\Sigma_k^2 \in \bbR^{2kn \times 2kn}$ are given by
    {\small
    \begin{align}
        K_k^{2} ={}& -\rlpar{R^2 + \rlpar{\overline{B}_k^2}^\top P_k^{2} \overline{B}_k^2}^{-1} \rlpar{\overline{B}_k^2}^\top P_k^{2} \overline{A}_k^2, \label{eqn:K2_k} \\
        \begin{split}\label{eqn:P2_k}
            P_k^{2} ={}& \overline{Q}_k^2 + \rlpar{\overline{A}_k^2}^\top P_k^{2} \overline{A}_k^2 - \rlpar{\overline{A}_k^2}^\top P_k^{2} \overline{B}_k^2 \\ & \rlpar{R^2 + \rlpar{\overline{B}_k^2}^\top P_k^{2} \overline{B}_k^2}^{-1} \rlpar{\overline{B}_k^2}^\top P_k^{2} \overline{A}_k^2,
        \end{split} \\
        L_k^{2} ={}& \overline{A}_k^2 \Sigma_k^{2} \rlpar{\overline{C}_k^2}^\top \rlpar{V^2 + \overline{C}_k^2 \Sigma_k^{2} \rlpar{\overline{C}_k^2}^\top}^{-1}, \label{eqn:L2_k}\\
        \begin{split}\label{eqn:Sigma2_k}
            \Sigma_k^{2} ={}& F_k^2 \overline{W}^2 \rlpar{F_k^2}^\top + \overline{A}_k^2 \Sigma_k^{2} \rlpar{\overline{A}_k^2}^\top - \overline{A}_k^2 \Sigma_k^{2} \rlpar{\overline{C}_k^2}^\top \\
            & \rlpar{V^2 + \overline{C}_k^2 \Sigma_k^2 \rlpar{\overline{C}_k^2}^\top}^{-1} \overline{C}_k^2 \Sigma_k^{2} \rlpar{\overline{A}_k^2}^\top.
        \end{split}
    \end{align}
    }
    Note that the maximizer's cost is bounded only when a solution $P_k^2$ exists with $R^2 + (\overline{B}_k^2)^\top P_k^{2} \overline{B}_k^2 \prec 0$, which requires that the maximizer input penalty $R^2$ is sufficiently negative definite.
    The closed-loop state and belief state dynamics are
    {\small
    \begin{align*} \nonumber
         \begin{bmatrix}
             X^2_{t+1,k} \\ z_{t+1,k}^2
         \end{bmatrix} &= \underbrace{\begin{bmatrix}
            \overline{A}_k^2 & \overline{B}_k^2 K_k^{2} \\
            L_k^{2} \overline{C}_k^2 & \overline{A}_k^2 + \overline{B}_k^2 K_k^{2} - L_k^{2} \overline{C}_k^2
        \end{bmatrix}}_{:= \tilde A_k^2} \begin{bmatrix}
             X^2_{t,k} \\ z_{t,k}^2
         \end{bmatrix} \\ &+ \underbrace{\begin{bmatrix}
            F_k^2 & 0 \\
            0 & L_k^2
        \end{bmatrix}}_{\tilde F_k^2} \begin{bmatrix}
             \tilde w^2_{t} \\ v_{t}^2
         \end{bmatrix}.
    \end{align*}
    }
    The optimal cost of the maximizer at each iteration is then
    \begin{align*}
        J_k^{2*} = Tr \rlpar{\tilde P_k^2 \tilde W_k^2} =  Tr \rlpar{\tilde \Sigma_k^2 \tilde Q_k^2},
    \end{align*}
    where {\small$$\tilde W_k^2 =  \tilde F_k^2 \begin{bmatrix}
        \overline{W}^2 & 0 \\
        0 & V^2
    \end{bmatrix} \rlpar{\tilde F_k^2}^\top, \ \tilde Q_k^2 =  \begin{bmatrix}
        \overline{Q}_k^2 & 0 \\
        0 & \rlpar{K_k^{2}}^\top R^2 K_k^{2}
    \end{bmatrix},$$} and $\tilde P_k^2$ and $\tilde \Sigma_k^2$ solve the respective Lyapunov equations
    {\small
    \begin{align*}
        \tilde P_k^2 &= \rlpar{\tilde A_k^2}^\top \tilde P_k^2 \rlpar{\tilde A_k^2} + \tilde Q_k^2, \ \tilde \Sigma_k^2 = \rlpar{\tilde A_k^2} \tilde \Sigma_k^2 \rlpar{\tilde A_k^2}^\top + \tilde W_k^2. 
    \end{align*}}

The internal state dimensions of players' feedback strategies grow towards infinity as the best response dynamics evolve. However, our numerical experiments show that the game's value converges rapidly within a few iterations. This suggests that strategies associated with increasingly higher-order belief states eventually provide no benefit in optimizing players' objective functions. To help explain this phenomenon, we will analyze the controllability and observability Gramian eigenvalues, along with the Hankel singular values of each player's augmented belief state dynamics in the next section.

\section{Controllability and Observability Metrics for Higher-Order Beliefs}\label{sec:control_observe_evaluation}
In this section, we first define the controllability and observability Gramians and Hankel singular values at each best response iteration. We then derive a result using Cholesky estimates to bound both the decay rates of these values and the error of low-order approximations of higher-order belief dynamics under certain conditions.

\subsection{Gramians and Hankel Singular Values}
With each iteration of the best response, the augmented belief dynamics that each player must observe and control become increasingly large-scale. We aim to evaluate how difficult it is for each player to control and observe their augmented belief dynamics. One approach is to analyze the controllability and observability Gramians of the augmented belief dynamics. Eigenvectors of Gramians associated with small eigenvalues define directions in the state space that are difficult to control and observe, while large eigenvalues indicate the opposite. These Gramians can be computed by solving the Lyapunov equations
\begin{align}
    W^i_{c,k} - \overline{A}_k^i W^i_{c,k} \left( \overline{A}_k^i  \right)^\top  &= \overline{B}_k^i \left( \overline{B}_k^i \right)^\top, \label{eqn:cntr_gramian} \\
    W^i_{o,k} - \left( \overline{A}_k^i  \right)^\top W^i_{o,k} \overline{A}_k^i &= \left( \overline{C}_k^i \right)^\top \overline{C}_k^i, \label{eqn:obsv_gramian}
\end{align}
where $W^i_{c,k} \succ 0$ and $W^i_{o,k} \succ 0$ are the controllability and observability Gramians for player $i$ ($i=1,2$) at iteration $k$. 

Another common approach to quantitatively analyze a dynamic system's controllability and observability from the perspective of each state's energy is through Hankel singular values. These singular values are closely related to the above Gramian eigenvalues \cite{boyd1994linear}. Large Hankel singular values correspond to state directions that are both easy to control and easy to observe, while small singular values indicate the opposite. The Hankel singular values of player $i$ ($i=1,2$) at iteration $k$ can be computed as the square roots of the eigenvalues of the product of the Gramians $W^i_{c,k}$ and $W^i_{o,k}$ in \eqref{eqn:cntr_gramian} and \eqref{eqn:obsv_gramian}:
\begin{equation}
    \sigma_{j,k}^i = \sqrt{\lambda_j \left( W^i_{c,k} W^i_{o,k} \right)}, \quad j = 1, 2, \ldots, n_k^i,
\end{equation}
where $n_k^i$ is the dimension of player $i$'s belief state at iteration $k$.

\subsection{Cholesky Estimates of Gramian Eigenvalue Decay Rates}
In many engineering applications \cite{pasqualetti2014controllability, ganapathy2021performance, antoulas2002decay}, the Gramian eigenvalues and Hankel singular values of large-scale dynamic systems often decay rapidly. This phenomenon is closely related to model reduction by means of balanced truncation. Thus, the Gramians of these dynamic systems may be approximated by low-rank matrices. We aim to quantify the decay rates of Gramian eigenvalues and Hankel singular values of both players' higher-order belief dynamics of the best responses in Section \ref{sec:best_response_dynamics}. Relatedly, we aim to bound the error of a low-order approximation of the higher-order belief dynamics. Following the approach in \cite{antoulas2002decay}, we now establish a theorem for this bound, which quantifies decay rates and approximation error bounds using a Cholesky factorization of a Cauchy matrix built from the eigenvalues of the higher-order belief dynamics matrices.

These Cholesky factors have been proven to effectively approximate the eigenvalue decay rates of continuous-time controllability and observability Gramians using the spectrum of the dynamics matrix. Moreover, higher-order belief dynamics with rapidly decaying Gramian eigenvalues can be approximated by low-order belief dynamics, with bounded error \cite{antoulas2002decay}. For discrete-time systems, the corresponding Cholesky factors can also be utilized for quantifying both the decay rates and approximation error bounds. The discrete-time Cholesky factors are defined by
{\small
\begin{equation}\label{eqn:cholesky_factor}
    \delta_{l,k}^i = \frac{-1}{2 Re\left( \frac{\lambda_{l,k}^i - 1}{\lambda_{l,k}^i + 1} \right)} \prod_{j=1}^{l-1} \left\vert \frac{(\lambda_{l,k}^i - \lambda_{j,k}^i)(\lambda_{l,k}^{i*} + 1)}{(\lambda_{l,k}^{i*}\lambda_{j,k}^i - 1)(\lambda_{l,k}^i + 1)} \right\vert^2,
\end{equation}}
where $\lambda_{j,k}^i \in \lambda(\overline{A}_k^i)$, $j = 1,2,\ldots,n_{k}^i$. The superscript $*$ denotes complex conjugation. We define an ordering with the $l$-th eigenvalue $\lambda_{l,k}^i$ selected according to
{\small
\begin{equation}
    \lambda_{l,k}^i = \underset{\lambda_k^i}{\text{argmax}} \left\{ \frac{-1}{2Re\left( \frac{\lambda_k^i - 1}{\lambda_k^i + 1} \right)} \prod_{j=1}^{l-1}  \left\vert \frac{(\lambda_{k}^i - \lambda_{j,k}^i)(\lambda_{k}^{i*} + 1)}{(\lambda_{k}^{i*}\lambda_{j,k}^i - 1)(\lambda_{k}^i + 1)} \right\vert^2 
    \right\},
\end{equation}}
where $\lambda_k^i \in \lambda(\overline{A}_k^i)/\mathcal{Z}_{l-1,k}^i$, and $\mathcal{Z}_{l-1,k}^i\coloneqq\{ \lambda_{j,k}^i : 1 \leq j \leq l - 1 \} \subset \lambda(\overline{A}_k^i)$ denotes the first $l-1$ selected eigenvalues.

Before presenting the main theorem, we define a matrix
\begin{equation} \label{cauchy}
    C_{hj,k}^i \coloneqq \frac{-(\lambda_{h,k}^i + 1)(\lambda_{j,k}^{i*} + 1)}{2(\lambda_{h,k}^i \lambda_{j,k}^{i*} - 1)},
\end{equation}
obtained by constructing a Cauchy matrix $C(\frac{\lambda_{k}^i - 1}{\lambda_{k}^i + 1}, \frac{\lambda_{k}^i - 1}{\lambda_{k}^i + 1})$
 defined in section 2, equation (3) of \cite{antoulas2002decay}. Furthermore, in Lemma 3.2 of \cite{antoulas2002decay}, it was proved that this Cauchy matrix is Hermitian positive definite, which can be factorized as $C_k^i = L_k^i \Delta_k^i L_k^{i*}$, and further the diagonal entries of $\Delta_k^i$ are given by \eqref{eqn:cholesky_factor}, with $L_k^i$ satisfying $\| L_k^i e_j \|_\infty = 1$ ($e_j$ is a standard unit vector). Let $X_k^i$ denote the matrix of right eigenvectors of $\overline{A}_k^i$, where each column of $X_k^i$ has unit norm. We further define a matrix 
\begin{equation*}
    Z_{j,k}^i \coloneqq [X_{1,k}^i L_k^i e_j, X_{2,k}^i L_k^i e_j, \ldots, X_{m_i,k}^i L_k^i e_j],
\end{equation*}
 where $X_{p,k}^i = X_k^i \text{diag} ((X_k^i)^{-1} \tilde b_{p,k}^i)$ and $\tilde b_{p,k}^i$ is the $p$-th column of $\tilde{B}_k^i = \sqrt{2} (\overline{A}_k^i + I) \overline{B}_k^i$, for $p = 1,2,\ldots,m_i$.

\begin{theorem}\label{thm1}
    Suppose that for each player $i$ ($i=1,2$), the augmented matrices $\overline{A}_k^i$ are stable and diagonalizable at iteration $k$. Further suppose that the pairs ($\overline{A}_k^i, \overline{B}_k^i$) is controllable and the pairs ($\overline{A}_k^i, \overline{C}_k^i$) is observable. Let $\hat{W}_{c_{lm_i},k}^i = \sum_{j=1}^l \delta_{j,k}^i Z_{j,k}^i Z_{j,k}^{i*}$. 
    If $\delta_{l,k}^i / \delta_{1,k}^i < \epsilon$, then $\hat{W}_{c_{lm_i},k}^i$ has rank at most $lm_i$ and satisfies
    {\small
    \begin{equation*}
        \| W_{c,k}^i - \hat{W}_{c_{lm_i},k}^i \|_\infty \leq \epsilon \delta_{1,k}^i (m_i n_{k}^i) (n_{k}^i - l) \left( \kappa_\infty (X_k^i) \| \tilde{B}_k^i \|_1 \right)^2, 
    \end{equation*}
    }
    where $\kappa(\cdot)$ denotes the condition number of a maitrx, and $\delta_{1,k}^i \approx \| W_{c,k}^i \|_2$.
\end{theorem} 
\begin{proof}
    We first transform matrices $\overline{A}_k^i$ and $\overline{B}_k^i$ to continuous-time counterparts $\tilde A_k^i = (I + \overline{A}_k^i)(\overline{A}_k^i - I)$ and $\tilde B_k^i = \sqrt{2} (\overline{A}_k^i + I) \overline{B}_k^i$ using the bilinear transformation from Section 2.2 of \cite{gloverAll1984}. This transformation ensures that the continous-time controllability Gramian of $(\tilde A_k^i, \tilde B_k^i)$ is identical to the discrete-time Gramian of $(\overline{A}_k^i, \overline{B}_k^i)$. Then we apply Theorem 3.2 of \cite{antoulas2002decay}, transforming the discrete-time eigenvalues $\lambda_{k}^i$ to continuous-time counterparts $\lambda_{c,k}^i$ through a transformation $\lambda_{c,k}^i = \frac{\lambda_{k}^i - 1}{\lambda_{k}^i + 1}$ for constructing the Cauchy matrix \eqref{cauchy}.
\end{proof}

Theorem \ref{thm1} shows that the controllability Gramians of both players' higher-order belief dynamics can be approximated by low-rank Gramians. The higher-order belief dynamic become increasingly difficult to control. Analogous statements are easily obtained for the observability Gramian, and related bounds for the Hankel singular values can be obtained as in \cite{antoulas2002decay}. Consequently, the strategies based on low-order belief states provide a good approximation of the Nash equilibrium strategies that correspond to infinite-dimensional belief states. 

\begin{figure}[tb]
    \centering
    \includegraphics[width=0.8\linewidth]{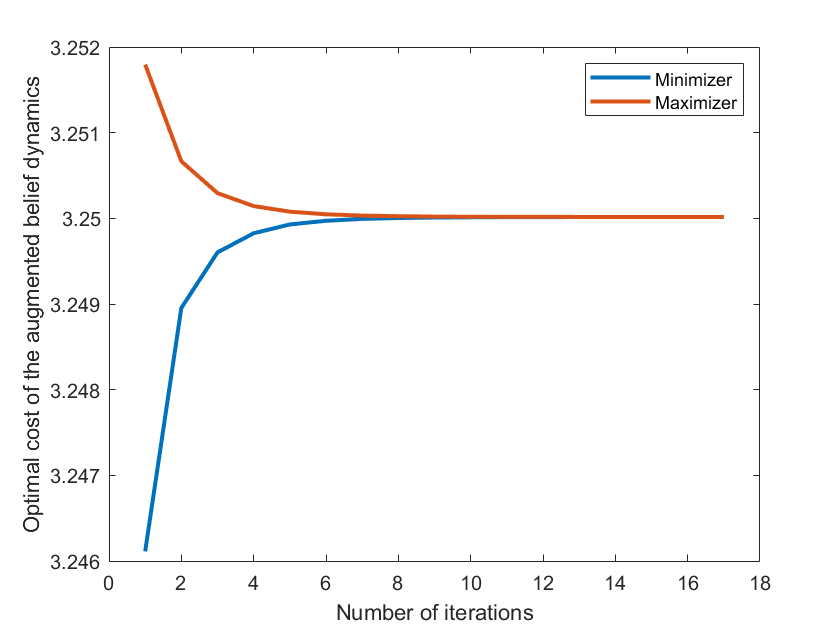}
    \caption{Optimal cost vs. best response iteration. The optimal cost converges within a few best response iterations.}
    \label{fig:convergence_IBR}
\end{figure}

\section{Numerical Experiments}\label{sec:experimental_results}
In this section, we first present the convergence performance of the best response dynamics. Next, we analyze the decay rates of the both players' controllability and observability Gramian eigenvalues and Hankel singular values. We then compare the Cholesky estimates with the actual eigenvalue decay rates to validate Theorem \ref{thm1}. Finally, we examine the decay rates of the Gramian eigenvalues and the Hankel singular values with an additional set of $1000$ randomly generated open-loop stable systems.

To illustrate the convergence performance of the best response dynamics, we consider a system for the numerical results in Section \ref{sec:opt_cost_converge} and \ref{sec:eigen_rapid_decay} with model parameters
\begin{equation*}
    A = 
    \begin{bmatrix}
        -0.3063 & -0.3580 \\
        0.5575 & -0.5273
    \end{bmatrix}, \quad
    B^1 = B^2 = 
    \begin{bmatrix}
        1 \\
        1
    \end{bmatrix},
\end{equation*}
\begin{equation*}
    C^1 = C^2 =
    \begin{bmatrix}
        1 & 1
    \end{bmatrix}, \quad
    W = I_2, \quad
    V^1 = V^2 = 1,
\end{equation*}
\begin{equation*}
    Q = I_2, \quad
    R^1 = 1, \quad R^2 = -7.5.
\end{equation*}
The value for penalty $R^2$ ensures that the upper value of the game is bounded. 
\subsection{Optimal Cost Convergence}\label{sec:opt_cost_converge}
Figure \ref{fig:convergence_IBR} shows the evolution of the optimal costs for both the maximizer and minimizer for their augmented belief dynamics at each best response iteration.
As the number of iterations increases, each player's optimal cost gradually converges to a constant value. This indicates that both players' best response strategies no longer improve their individual optimal cost. This point of convergence corresponds to a Nash equilibrium, where neither player obtains a benefit by deviating from their current strategies. The next subsection delves deeper into the mechanisms behind this convergence by analyzing the decay rates of the controllability and observability Gramian eigenvalues and Hankel singular values.
\vspace{-15pt}
\begin{figure}[htbp] 
    \centering
    \begin{subfigure}[b]{\linewidth}
        \centering
        \includegraphics[width=0.8\linewidth]{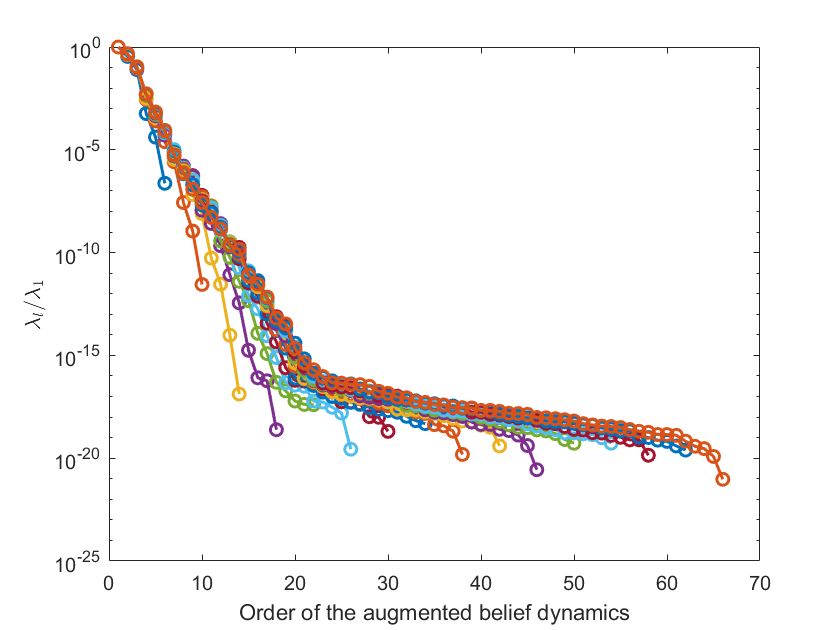}
        \caption{Rapid decay rates of the minimizer's Gramian eigenvalues.}
        \label{fig:wc_min}
    \end{subfigure}
    \vfill
    \begin{subfigure}[b]{\linewidth}
        \centering
        \includegraphics[width=0.8\linewidth]{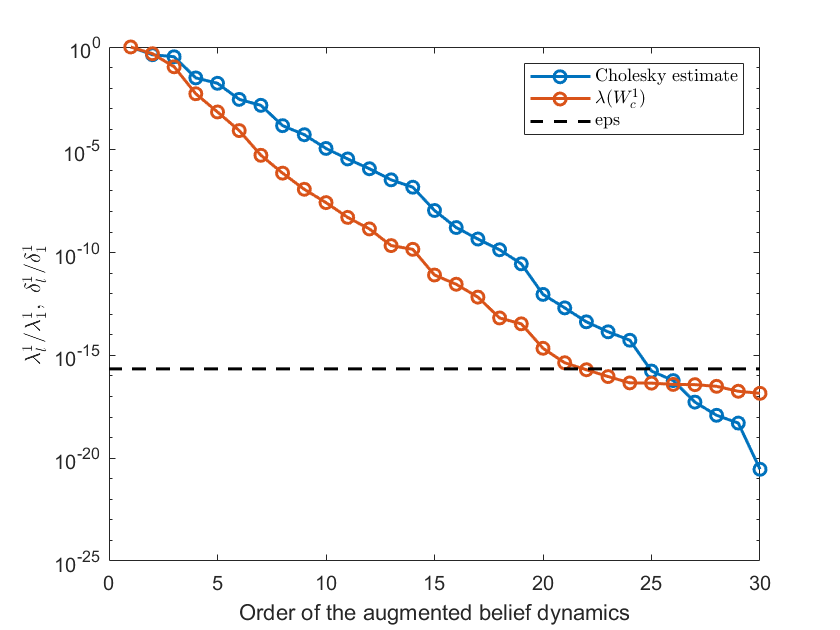}
        \caption{Comparison of the Cholesky estimates $\delta_{l}^1/\delta_{1}^1$ and discrete-time Gramian eigenvalue ratios $\lambda_l(W_c^1)/\lambda_1(W_c^1)$.}
        \label{fig:choleksy_est}
    \end{subfigure}
    \caption{The eigenvalues of the minimizer's controllability Gramian decay exponentially (each series represents the magnitudes of these eigenvalues for each best response iteration). The Cholesky ratios provide a good approximation to the actual Gramian eigenvalue decay rates.}
    \label{fig:Gramians}
\end{figure}
\vspace{-10pt}
\subsection{Rapid Decay Rates of Gramian Eigenvalues and Hankel Singular Values}\label{sec:eigen_rapid_decay}
Figure \ref{fig:wc_min} shows the rapid decay of the minimizer's controllability Gramian eigenvalues at each iteration (similar phenomena also occur in the maximizer's controllability Gramian eigenvalues and their Hankel singular values). This indicates that their higher-order augmented belief dynamics become increasingly difficult to control and observe. The use of higher-order beliefs yields diminishing and ultimately \emph{vanishing} returns for both the minimizer and maximizer. Consequently, their best response strategies converge toward a Nash equilibrium. Figure \ref{fig:choleksy_est} shows that the Cholesky estimates provide a good approximation of the decay rates of the minimizer's controllability Gramian eigenvalues by using only the spectrum of $\overline{A}_k^i$. Due to the rapid decay of each player's higher-order Gramian eigenvalues, higher-order belief dynamics can be approximated by lower-order belief dynamics with a bounded approximation error.

\subsection{Best Response Convergence for Many Problem Instances}\label{sec:extral_instances}

We performed numerical experiments across $1000$ randomly generated problem instances, where system parameters were independently drawn from a standard normal distribution. The matrix $A$ was scaled to be open-loop stable, and the penalty $R^2$ was chosen sufficiently large to guarantee a bounded upper value for the game. In all cases, the results support the same observations described above: rapid convergence of the optimal cost of each player to an equilibrium value and exponential decay rates for Gramian eigenvalues and Hankel singular values of the augmented belief dynamics. As shown in Table \ref{tab:table_eig}, approximately $75\%$ of these values have magnitudes less than $10^{-5}$, and $20-60\%$ fall below $10^{-10}$. These findings may suggest a general phenomenon in dynamic games with partial and asymmetric information, where low-order beliefs and strategies with limited internal dimension provide near-equilibrium performance.

\begin{table}[tb]
    \centering
    \caption{Proportion of each player's Gramian eigenvalues and Hankel singular values (HSV) less than $ \{10^{-5}, 10^{-10} \}$ for $1000$ random systems ($n = m_1 = m_2 = p_1 = p_2 = 1$), each entry in the table shows the percentage of eigenvalues/singular values with magnitudes below a threshold).}
    \label{tab:table_eig}
    \renewcommand{\arraystretch}{1.2}
    \begin{tabular}{|c|c|c|}
        \hline
         & $\lambda_l,\sigma_l \leq 10^{-5}$ & $\lambda_l,\sigma_l \leq 10^{-10}$ \\
         \hline
        $W_{c,\text{ave}}^1$ & $73.67\%$ & $54.67\%$ \\
        $W_{o,\text{ave}}^1$ & $75.12\%$ & $55.79\%$ \\
        $W_{c,\text{ave}}^2$ & $73.99\%$ & $55.65\%$ \\
        $W_{o,\text{ave}}^2$ & $76.08\%$ & $56.22\%$ \\
        $HSV_{\text{ave}}^1$ & $78.21\%$ & $20.43\%$ \\
        $HSV_{\text{ave}}^2$ & $77.93\%$ & $36.65\%$ \\
        \hline
    \end{tabular}
\end{table}

\section{Conclusion}\label{sec:conclusion}
In this work, we derived explicit best response expressions for each player in a zero-sum stochastic LQDG and analyzed the controllability and observability of their higher-order belief dynamics through extensive numerical experiments. Future work could extend these best response dynamics to $N$-player nonzero-sum dynamic games, as well as to games with nonlinear, non-symmetric, or time-varying dynamics.


\bibliographystyle{IEEEtran}
\bibliography{bibliography.bib}

\begin{thebibliography}{10}
\providecommand{\url}[1]{#1}
\csname url@rmstyle\endcsname
\providecommand{\newblock}{\relax}
\providecommand{\bibinfo}[2]{#2}
\providecommand\BIBentrySTDinterwordspacing{\spaceskip=0pt\relax}
\providecommand\BIBentryALTinterwordstretchfactor{4}
\providecommand\BIBentryALTinterwordspacing{\spaceskip=\fontdimen2\font plus
\BIBentryALTinterwordstretchfactor\fontdimen3\font minus
  \fontdimen4\font\relax}
\providecommand\BIBforeignlanguage[2]{{%
\expandafter\ifx\csname l@#1\endcsname\relax
\typeout{** WARNING: IEEEtran.bst: No hyphenation pattern has been}%
\typeout{** loaded for the language `#1'. Using the pattern for}%
\typeout{** the default language instead.}%
\else
\language=\csname l@#1\endcsname
\fi
#2}}

\bibitem{basar1998}
T.~Başar and G.~J. Olsder, \emph{Dynamic Noncooperative Game Theory, 2nd
  Edition}.\hskip 1em plus 0.5em minus 0.4em\relax Society for Industrial and
  Applied Mathematics, 1998.

\bibitem{tamer1973multistage}
B.~Tamer and M.~Max, ``A multistage pursuit-evasion game that admits a gaussian
  random process as a maximin control policy,'' \emph{Stochastics: An
  International Journal of Probability and Stochastic Processes}, vol.~1, no.
  1-4, pp. 25--69, 1973.

\bibitem{rhodes1969differential}
I.~Rhodes and D.~Luenberger, ``Differential games with imperfect state
  information,'' \emph{IEEE Transactions on Automatic Control}, vol.~14, no.~1,
  pp. 29--38, 1969.

\bibitem{zheng2013decomposition}
J.~Zheng and D.~A. Casta{\~n}{\'o}n, ``Decomposition techniques for markov
  zero-sum games with nested information,'' in \emph{52nd IEEE conference on
  decision and control}.\hskip 1em plus 0.5em minus 0.4em\relax IEEE, 2013, pp.
  574--581.

\bibitem{gupta2014common}
A.~Gupta, A.~Nayyar, C.~Langbort, and T.~Basar, ``Common information based
  markov perfect equilibria for linear-gaussian games with asymmetric
  information,'' \emph{SIAM Journal on Control and Optimization}, vol.~52,
  no.~5, pp. 3228--3260, 2014.

\bibitem{kartik2019zero}
D.~Kartik and A.~Nayyar, ``Zero-sum stochastic games with asymmetric
  information,'' in \emph{2019 IEEE 58th Conference on Decision and Control
  (CDC)}.\hskip 1em plus 0.5em minus 0.4em\relax IEEE, 2019, pp. 4061--4066.

\bibitem{hambly2023linear}
B.~Hambly, R.~Xu, and H.~Yang, ``Linear-quadratic gaussian games with
  asymmetric information: Belief corrections using the opponents actions,''
  \emph{arXiv preprint arXiv:2307.15842}, 2023.

\bibitem{schwarting2021stochastic}
W.~Schwarting, A.~Pierson, S.~Karaman, and D.~Rus, ``Stochastic dynamic games
  in belief space,'' \emph{IEEE Transactions on Robotics}, vol.~37, no.~6, pp.
  2157--2172, 2021.

\bibitem{peters2022learning}
L.~Peters, D.~Fridovich-Keil, L.~Ferranti, C.~Stachniss, J.~Alonso-Mora, and
  F.~Laine, ``Learning mixed strategies in trajectory games,'' \emph{arXiv
  preprint arXiv:2205.00291}, 2022.

\bibitem{bichler2023learning}
M.~Bichler, N.~Kohring, and S.~Heidekr{\"u}ger, ``Learning equilibria in
  asymmetric auction games,'' \emph{INFORMS Journal on Computing}, vol.~35,
  no.~3, pp. 523--542, 2023.

\bibitem{gupta2016dynamic}
A.~Gupta, C.~Langbort, and T.~Ba{\c{s}}ar, ``Dynamic games with asymmetric
  information and resource constrained players with applications to security of
  cyberphysical systems,'' \emph{IEEE Transactions on Control of Network
  Systems}, vol.~4, no.~1, pp. 71--81, 2016.

\bibitem{huang2020dynamic}
Y.~Huang, J.~Chen, L.~Huang, and Q.~Zhu, ``Dynamic games for secure and
  resilient control system design,'' \emph{National Science Review}, vol.~7,
  no.~7, pp. 1125--1141, 2020.

\bibitem{harsanyi1967games}
J.~C. Harsanyi, ``Games with incomplete information played by “bayesian”
  players, i--iii part i. the basic model,'' \emph{Management science},
  vol.~14, no.~3, pp. 159--182, 1967.

\bibitem{akerlof1970market}
G.~A. Akerlof, ``The market for “lemons”: Quality uncertainty and the
  market mechanism,'' \emph{The quarterly journal of economics}, vol.~84,
  no.~3, pp. 488--500, 1970.

\bibitem{aumann1976agreeing}
R.~J. Aumann, ``Agreeing to disagree,'' \emph{The Annals of Statistics},
  vol.~4, no.~6, pp. 1236--1239, 1976.

\bibitem{spence1978job}
M.~Spence, ``Job market signaling,'' in \emph{Uncertainty in economics}.\hskip
  1em plus 0.5em minus 0.4em\relax Elsevier, 1978, pp. 281--306.

\bibitem{stiglitz1981credit}
J.~E. Stiglitz and A.~Weiss, ``Credit rationing in markets with imperfect
  information,'' \emph{The American economic review}, vol.~71, no.~3, pp.
  393--410, 1981.

\bibitem{fudenberg1998theory}
D.~Fudenberg and D.~Levine, ``The theory of learning in games, economics
  learning and social evolution series,'' 1998.

\bibitem{reeves2012computing}
D.~Reeves and M.~P. Wellman, ``Computing best-response strategies in infinite
  games of incomplete information,'' \emph{arXiv preprint arXiv:1207.4171},
  2012.

\bibitem{harris1998rate}
C.~Harris, ``On the rate of convergence of continuous-time fictitious play,''
  \emph{Games and Economic Behavior}, vol.~22, no.~2, pp. 238--259, 1998.

\bibitem{hofbauer2006best}
J.~Hofbauer and S.~Sorin, ``Best response dynamics for continuous zero-sum
  games,'' \emph{Discrete and Continuous Dynamical Systems Series B}, vol.~6,
  no.~1, p. 215, 2006.

\bibitem{sandholm2010population}
W.~H. Sandholm, \emph{Population games and evolutionary dynamics}.\hskip 1em
  plus 0.5em minus 0.4em\relax MIT press, 2010.

\bibitem{williams2018best}
G.~Williams, B.~Goldfain, P.~Drews, J.~M. Rehg, and E.~A. Theodorou, ``Best
  response model predictive control for agile interactions between autonomous
  ground vehicles,'' in \emph{2018 IEEE International Conference on Robotics
  and Automation (ICRA)}.\hskip 1em plus 0.5em minus 0.4em\relax IEEE, 2018,
  pp. 2403--2410.

\bibitem{wang2019game}
Z.~Wang, R.~Spica, and M.~Schwager, ``Game theoretic motion planning for
  multi-robot racing,'' in \emph{Distributed Autonomous Robotic Systems: The
  14th International Symposium}.\hskip 1em plus 0.5em minus 0.4em\relax
  Springer, 2019, pp. 225--238.

\bibitem{wang2021game}
M.~Wang, Z.~Wang, J.~Talbot, J.~C. Gerdes, and M.~Schwager, ``Game-theoretic
  planning for self-driving cars in multivehicle competitive scenarios,''
  \emph{IEEE Transactions on Robotics}, vol.~37, no.~4, pp. 1313--1325, 2021.

\bibitem{pasqualetti2014controllability}
F.~Pasqualetti, S.~Zampieri, and F.~Bullo, ``Controllability metrics,
  limitations and algorithms for complex networks,'' \emph{IEEE Transactions on
  Control of Network Systems}, vol.~1, no.~1, pp. 40--52, 2014.

\bibitem{ganapathy2021performance}
K.~Ganapathy, J.~Ruths, and T.~Summers, ``Performance bounds for optimal and
  robust feedback control in networks,'' \emph{IEEE Transactions on Control of
  Network Systems}, vol.~8, no.~4, pp. 1754--1766, 2021.

\bibitem{antoulas2002decay}
A.~C. Antoulas, D.~C. Sorensen, and Y.~Zhou, ``On the decay rate of hankel
  singular values and related issues,'' \emph{Systems \& Control Letters},
  vol.~46, no.~5, pp. 323--342, 2002.

\bibitem{willemsLeast1971}
J.~Willems, ``Least squares stationary optimal control and the algebraic
  {{Riccati}} equation,'' \emph{IEEE Transactions on Automatic Control},
  vol.~16, no.~6, pp. 621--634, Dec. 1971.

\bibitem{molinariStable1973a}
B.~Molinari, ``The stable regulator problem and its inverse,'' \emph{IEEE
  Trans. on Automatic Control}, vol.~18, no.~5, pp. 454--459, Oct. 1973.

\bibitem{boyd1994linear}
S.~Boyd, L.~El~Ghaoui, E.~Feron, and V.~Balakrishnan, \emph{Linear matrix
  inequalities in system and control theory}.\hskip 1em plus 0.5em minus
  0.4em\relax SIAM, 1994.

\bibitem{gloverAll1984}
{\relax Keith}.~Glover, ``All optimal {{Hankel-norm}} approximations of linear
  multivariable systems and their {{L}}, {$\infty$} -error bounds{\dag},''
  \emph{International Journal of Control}, vol.~39, no.~6, pp. 1115--1193, June
  1984.

\end{thebibliography}

\end{document}